\documentclass[final]{article}

\usepackage{times}
\usepackage{tikz}
\usetikzlibrary{arrows,snakes,backgrounds}
\usepackage[a4paper]{geometry}
\usepackage{stackrel}
\usepackage{hyperref}
\usepackage{ucs}
\usepackage[latin1]{inputenc}
\usepackage{amsmath}
\usepackage{amsthm}
\usepackage{amssymb,amsfonts,amsxtra}
\usepackage{stmaryrd}
\usepackage[mathcal,mathscr]{euscript}
\usepackage[english]{babel}
\usepackage{afterpage}
\usepackage{newvbtm}
\makeatletter
\newif\if@restonecol
\makeatother

\usepackage[onelanguage,boxed,titlenotnumbered,figure,commentsnumbered]{algorithm2e}

\newtheorem{theorem}{Theorem}[section]
\newtheorem{lemma}[theorem]{Lemma}

\def\ok#1{\mbox{\raisebox{0ex}[1ex][1ex]{$#1$}}}
\def \tuple#1{\langle #1 \rangle}

\newcommand{\sra}{{\shortrightarrow}}

\newcommand{\id}{\mathrm{id}}

\newcommand{\ud}{\triangleq}

\newcommand{\ra}{\rightarrow}
\newcommand{\raee}{\sra^{\exists}}
\newcommand{\nraee}{\not\!\!\sra^{\exists}}

\newcommand{\tre}{\trianglerighteq}
\newcommand{\tle}{\trianglelefteq}

\newcommand{\Lra}{\Leftrightarrow}
\newcommand{\Ra}{\Rightarrow}
\newcommand{\La}{\Leftarrow}

\newcommand{\cK}{{\mathcal{K}}}

\newcommand{\cP}{{\mathcal{P}}}

\newcommand{\Psim}{\ensuremath{P_{\mathrm{sim}}}}

\newcommand{\Rsim}{\ensuremath{R_{\mathrm{sim}}}}

\newcommand{\ESim}{\ensuremath{\mathrm{ESim}}}

\newcommand{\tl}{\vartriangleleft}

\newcommand{\sraee}{\sra^{\exists}}

\DeclareMathOperator{\Rel}{Rel}

\def\ctemplate#1#2{#1$\langle$#2$\rangle$}
\def\clist#1{\ctemplate{\KwSty{list}}{\KwSty{#1}}\ }

\DeclareMathOperator{\PStable}{PStable}
\DeclareMathOperator{\RStable}{RStable}
\DeclareMathOperator{\PR}{PR}

\DeclareMathOperator{\RRefiner}{RRefiner}
\DeclareMathOperator{\PRefiner}{PRefiner}

\DeclareMathOperator{\pre}{pre}

\DeclareMathOperator{\post}{post}

\DeclareMathOperator{\Part}{Part}

\DeclareMathOperator{\Split}{\mathit{Split}}

\DeclareMathOperator{\RT}{RT}
\DeclareMathOperator{\GPP}{GPP}

\DeclareMathOperator{\parent}{parent}

\DeclareMathOperator{\Remove}{\mathit{Rem}}
\DeclareMathOperator{\oldRemove}{\mathit{oldRem}}

\DeclareMathOperator{\Count}{\mathrm{Count}}

\DeclareMathOperator{\RStabilize}{\mathit{RStabilize}}
\DeclareMathOperator{\PStabilize}{\mathit{PStabilize}}

\newcommand{\cstate}{\KwSty{State}}
\newcommand{\cblock}{\KwSty{Block}}

\newcommand{\creturn}{\KwSty{return}\ }

\newcommand{\ctt}{\KwSty{tt}}
\newcommand{\cff}{\KwSty{ff}}
\newcommand{\cnull}{\KwSty{null}}

\SetKwFor{ForAll}{forall}{do}{endfall}

\def\ctemplate#1#2{#1$\langle$#2$\rangle$}
\def\clist#1{\ctemplate{\KwSty{list}}{\KwSty{#1}}}

\begin{document}

\title{An Efficient Simulation Algorithm on Kripke Structures}
\author{\normalsize{\sc Francesco Ranzato}\\[10pt]
{\small Dipartimento di Matematica, 
University of Padova, Italy}}

\pagestyle{plain}
\date{}
\maketitle

\begin{abstract}
A number of algorithms for computing the simulation preorder (and equivalence) 
on Kripke structures 
are available.  Let $\Sigma$ denote the state space, $\sra$  the 
transition relation and $P_{\mathrm{sim}}$ the partition of  $\Sigma$ induced by 
simulation equivalence.
While some algorithms are designed to reach the best space bounds, 
whose dominating additive term
is $|P_{\mathrm{sim}}|^2$, other
algorithms are devised to attain the best time complexity $O(|P_{\mathrm{sim}}||\sra|)$. 
We present a novel
simulation algorithm
which is both space and time efficient: it runs in 
$O(|P_
{\mathrm{sim}}|^2 \log |\Psim| + |\Sigma|\log |\Sigma|)$  space and
$O(|P_{\mathrm{sim}}||\sra|\log |\Sigma|)$ 
time.  Our simulation algorithm thus reaches the best space bounds while closely 
approaching the best time complexity.  
\end{abstract}

\section{Introduction}
The simulation preorder is a 
fundamental behavioral relation
widely used in process algebra for establishing system correctness
and in model checking as a suitable abstraction for reducing the size
of state spaces \cite{cgp99}. 
The problem of efficiently computing the simulation preorder (and consequently simulation equivalence) 
on  finite 
Kripke structures has been thoroughly investigated and generated a 
number of simulation algorithms \cite{bloom89,bp95,bg03,cps93,CRT11,gpp03,GP08,hhk95,rt10,TC01}. 
Both time and space complexities play 
an important role 
in simulation algorithms, since in several applications,
especially in model checking,  
memory requirements may become a serious bottleneck as the input transition system grows.  

\paragraph*{\textbf{State of the Art.}}
Consider a finite Kripke structure where 
$\Sigma$ denotes
the state space, $\sra$ the transition relation and $P_{\mathrm{sim}}$
the partition of $\Sigma$ induced by simulation equivalence. 
The best simulation algorithms are those by, in chronological order, Gentilini,
Piazza and Policriti (GPP)~\cite{gpp02,gpp03} (subsequently corrected in~\cite{GP08}), 
Ranzato and Tapparo (RT)~\cite{rt07,rt10},  
Markovski (Mar)~\cite{mar11}, C{\'e}c{\'e} (Space-C{\'e}c and Time-C{\'e}c)~\cite{cec13}. 
The simulation algorithms GPP and RT are designed for Kripke structures, while  
Space-C{\'e}c, Time-C{\'e}c and Mar are for more general labeled transition systems. 
Their space and time complexities are summarized in the following table. 
\begin{center}
{\small 
\renewcommand{\arraystretch}{1.25}
\centering
\begin{tabular}{|l||l|l|}
\hline
 ~\textbf{Algorithm} & ~\textbf{Space complexity} & ~\textbf{Time complexity}\\
\hline \hline
~Space-C{\'e}c  \cite{cec13} & ~$O(|P_{\mathrm{sim}}|^2  + |\sra|\log|\sra|)$ & ~$O(|\Psim|^2|\sra|)$ \\[1mm]
~Time-C{\'e}c  \cite{cec13} & ~$O(|P_{\mathrm{sim}}||\Sigma|\log|\Sigma| + |\sra|\log |\sra|)$ & ~$O(|\Psim||\sra|)$ \\[1mm]
~$\GPP$ \cite{gpp03}& ~$O(|P_{\mathrm{sim}}|^2\log |\Psim| + |\Sigma|\log|\Sigma|)$  &~$O(|\Psim|^2|\sra|)$   \\[1mm]
~Mar \cite{mar11}& ~$O((|\Sigma|+|P_{\mathrm{sim}}|^2) \log |\Psim|)$  &~$O(|\sra| + |\Psim||\Sigma| +|\Psim|^3)$~   \\[1mm]
~$\RT$ \cite{rt10} & ~$O(|P_{\mathrm{sim}}||\Sigma|\log|\Sigma|)$ & ~$O(|\Psim||\sra|)$ \\[1mm]
\hline 
~$\ESim$ (this paper)~  & ~$O(|\Psim|^2\log|\Psim| + |\Sigma|\log |\Sigma|)$~ & ~$O(|P_{\mathrm{sim}}|
|\sra|\log |\Sigma|)$ \\
\hline
\end{tabular}
}
\end{center}
We remark that all the above space bounds are bit space complexities,  
i.e., the word size is a single bit. 
Let us also remark that all the articles \cite{gpp02,gpp03,GP08}
state that the bit space complexity of GPP is in $O(|P_{\mathrm{sim}}|^2  + |\Sigma| \log |\Psim|)$. However, as observed also in \cite{cec13}, this is not precise. 
In fact, the algorithm GPP \cite[Section~4, p.~98]{gpp03} 
assumes that the states belonging to some block 
are stored as a doubly linked list, and  this entails a
bit space complexity in $O(|\Sigma|\log|\Sigma|)$. Furthermore,  GPP 
uses Henzinger, Henzinger and Kopke \cite{hhk95} simulation 
algorithm (HKK) as a subroutine, whose 
bit space complexity 
is in $O(|\Sigma|^2 \log |\Sigma|)$,
which is 
called on a Kripke structure where states 
are blocks of the current partition. The bit space complexity of GPP
must therefore include an additive term 
$|\Psim|^2 \log |\Psim|$ and therefore
results to be $O(|P_{\mathrm{sim}}|^2 \log |P_{\mathrm{sim}}| + |\Sigma| \log |\Sigma|)$. 
It is worth observing that a 
space complexity in $O(|P_{\mathrm{sim}}|^2  + |\Sigma| \log |\Psim|)$
can be considered optimal for a simulation algorithm, since this 
is of the same order as the
size of the output, which needs $|P_{\mathrm{sim}}|^2$ space for
storing the simulation preorder as a partial order on simulation equivalence
classes and $|\Sigma| \log |\Psim|$ space for storing the simulation 
equivalence class for any state. Hence, the bit space complexities of 
GPP and Space-C{\'e}c can be considered quasi-optimal. 
As far as time complexity is concerned, 
the algorithms RT and Time-C{\'e}c both feature the best
time bound $O(|P_{\mathrm{sim}}| |\sra|)$.

\paragraph*{\textbf{Contributions.}}
We present here a novel space and time Efficient 
Simulation algorithm, called ESim, which features a 
 bit space complexity in 
$O(|P_{\mathrm{sim}}|^2 \log |\Psim| + |\Sigma|\log |\Sigma|)$ and a
time complexity in 
$O(|P_{\mathrm{sim}}||\sra|\log |\Sigma|)$. 
Thus, ESim reaches the best space bound of GPP and significantly improves the GPP time 
bound $O(|\Psim|^2|\sra|)$ 
by replacing a multiplicative factor $|\Psim|$ with $\log |\Sigma|$. Furthermore, 
ESim significantly 
improves the RT space bound $O(|P_{\mathrm{sim}}||\Sigma|\log|\Sigma|)$ and closely 
approaches the best time bound $O(|\Psim||\sra|)$ of RT and Time-C{\'e}c. 

ESim is a
partition refinement algorithm, meaning that it maintains and iteratively refines
a so-called partition-relation pair $\tuple{P,\tle}$, where
$P$ is a partition of $\Sigma$ that overapproximates the
final simulation partition $\Psim$, while $\tle$ is a binary relation over
$P$ which overapproximates the final simulation
preorder. ESim relies on the following three main points, which
in particular allow to attain the above
complexity bounds.
\begin{itemize}
\item[{\rm (1)}] Two distinct 
notions of partition and relation stability for a partition-relation
pair are introduced. Accordingly, at a logical level,   
ESim is designed as a partition refinement algorithm 
which iteratively performs two clearly distinct refinement steps:
the refinement 
of the current partition $P$ which splits some blocks of $P$ 
and the refinement of the relation $\tle$ which removes some pairs of blocks
from $\tle$.     
\item[{\rm (2)}] ESim exploits a logical characterization
of partition refiners, i.e.\ blocks of $P$ that allow to
split the current partition $P$, which admits an efficient implementation.  
\item[{\rm (3)}] ESim only relies on data structures, like lists and matrices, 
that are indexed on and contain blocks 
of the current partition $P$. The hard task here is to devise
efficient ways to  keep
updated these partition-based data structures along the iterations of ESim.  
We show that this can be done efficiently,  
in particular
by resorting to 
Hopcroft's ``process the smaller half'' principle \cite{hop71} when
updating a crucial data structure 
after a partition split. 
\end{itemize}

This is the full version of the conference paper \cite{ran-mfcs13}.

\section{Background}\label{background}

\paragraph*{\rm \textbf{Notation.}}
If $R\subseteq \Sigma\times \Sigma$ is any relation 
and $X\subseteq \Sigma$ then $R(X)\,\ok{\ud}\, \{x'\in
\Sigma~|~ \exists x\in X.\: \ok{(x,x')}\in R\}$. 
Recall that $R$ is a preorder relation when it is reflexive and
transitive. 
If $f$ is a function defined on $\wp(\Sigma)$ and $x\in \Sigma$ then
we often write $f(x)$ to mean $f(\{x\})$. 
$\Part(\Sigma)$ denotes the set of partitions of
$\Sigma$.  If $P\in \Part(\Sigma)$, $s\in \Sigma$ and $S\subseteq \Sigma$ 
then $P(s)$
denotes the block of $P$ that contains $s$ while $P(S) = \cup_{s\in S} P(s)$. 
$\Part(\Sigma)$ is endowed
with the standard partial order $\preceq$: $P_1 \preceq
P_2$, i.e.\ $P_2$ is coarser than $P_1$, iff for any $s\in \Sigma$, 
$P_1(s) \subseteq P_2(s)$. 
If $P_1\preceq P_2$ and $B\in P_1$ then $P_2(B)$ is a block of $P_2$ which is
also denoted by
$\parent_{P_2}(B)$. 
For a given nonempty subset $S\subseteq \Sigma$ called splitter, we
denote by $\Split(P,S)$ the partition obtained from $P$ by replacing
each block $B\in P$ with $B\cap S$ and
$B\smallsetminus S$, 
where we also allow no splitting, namely
$\Split(P,S)=P$ (this happens exactly when $P(S)=S$).

\paragraph*{\rm \textbf{Simulation Preorder and Equivalence.}}
A transition system $(\Sigma ,\sra)$ consists of a 
set $\Sigma$ of states and of a transition relation $\sra 
\subseteq \Sigma \times
\Sigma$. 
Given a set $\mathit{AP}$ of atoms
(of some specification language), a Kripke structure (KS) 
$\cK= (\Sigma ,\sra,\ell)$ over
$\mathit{AP}$ consists of a transition system $(\Sigma ,\sra)$
together with a state labeling function $\ell:\Sigma \ra
\wp(\mathit{AP})$. The state partition induced by $\ell$ is denoted by
$P_\ell \,\ok{\ud}\, \{\{s'\in
\Sigma~|~ \ell(s)=\ell(s')\}~|~s\in \Sigma\}$.
The predecessor/successor
transformers $\pre , \post:\wp(\Sigma)\ra \wp(\Sigma)$  are
defined as usual:
$\pre (T) \ud \{ s\in \Sigma ~|~\exists t\in T.\; s \sra
t\}$ and $\post (S) \ud \{ t\in \Sigma ~|~\exists s\in S.\; s \sra
t\}$.
If $S_1,S_2\subseteq \Sigma$ then $S_1
\sraee S_2$ iff there exist $s_1\in S_1$ and $s_2\in
S_2$ such that $s_1 \sra s_2$.

A relation 
 $R\subseteq \Sigma \times \Sigma $ is a
 simulation on a Kripke structure $(\Sigma,\sra,\ell)$ 
 if for any $s,s'\in \Sigma$, if $s'\in R(s)$ then: 
 \begin{itemize}
 \item[{\rm (A)}] $\ell(s) = \ell (s')$; 
 \item[{\rm (B)}] for
 any $t\in \Sigma $ such that $\ok{s \sra t}$, there exists $t'\in \Sigma$
 such that $\ok{s'\sra t'}$
 and $t'\in R(t)$.
 \end{itemize}
Given $s,t\in \Sigma$, $t$ simulates $s$, denoted by $s\leq t$, if there exists
a simulation relation $R$ such that $t\in R(s)$. 
It turns out that the largest simulation on a given KS exists,
is a preorder relation called simulation preorder and is
denoted by $R_{\mathrm{sim}}$. Thus, for any $s,t\in \Sigma$, we have that $s\leq
t$ iff $(s,t)\in  \ok{R_{\mathrm{sim}}}$.
The simulation partition $P_{\mathrm{sim}}\in \Part(\Sigma)$ is 
the symmetric reduction of 
$R_{\mathrm{sim}}$, namely, for any $s,t\in \Sigma$, 
$\Psim (s) = \Psim (t)$
iff $s\leq t$ and $t\leq s$.

\section{Logical Simulation Algorithm}
 \label{secbasicsa}

\subsection{Partition-Relation Pairs}
A partition-relation pair $\cP= \tuple{P,\tle}$, PR
for short, is a state partition $P\in \Part(\Sigma)$ together with
a binary relation $\tle\; \subseteq P\times P$ between blocks of $P$.
We write $B\tl C$ when $B\tle C$ and $B\neq C$ and $(B',C')\tle (B,C)$
when $B'\tle B$ and $C'\tle C$. When $\tle$ is a 
preorder/partial order  
then $\cP$ is called, respectively, a
preorder/partial order PR.

PRs allow to represent symbolically, i.e.\ 
through state partitions, a relation between states. 
A relation $R\subseteq \Sigma
\times \Sigma$ induces a PR $\PR(R)=\tuple{P,\tle}$ defined as follows: 
\begin{itemize}
\item[--] 
for any $s$, $P(s)\ud \{t\in \Sigma~|~R(s)=R(t)\}$;
\item[--] for any $s,t$, $P(s)\tle P(t) 
\text{~iff~} t\in R(s)$.
\end{itemize}
It is easy to note that if $R$ is a preorder then 
$\PR(R)$ is a partial order PR.  
On the other hand, a PR $\cP= \tuple{P,\tle}$ 
induces the following relation $\Rel(\cP)\subseteq \Sigma\times \Sigma$:
$$(s,t)\in \Rel(\cP) \;\Lra\; P(s) \tle P(t).$$
Here, if $\cP$ is a preorder PR then $\Rel(\cP)$ is clearly a preorder. 
  
A PR $\cP=\tuple{P,\tle}$ is defined to be a
simulation PR on a KS $\cK$ when $\Rel(\cP)$ is a 
simulation on $\cK$, namely when $\cP$ represents a simulation
relation between states. Hence, if $\cP$ is a simulation PR and
$P(s)=P(t)$ then $s$ and $t$ are simulation equivalent, while if
$P(s)\tle P(t)$ then $t$ simulates $s$.

Given a  PR $\cP= \tuple{P,\tle}$, the map
$\mu_\cP : \wp(\Sigma) \ra \wp(\Sigma)$ is defined as follows: 
$$\text{for any~} X\in
\wp(\Sigma),\: \mu_\cP (X) \ud \Rel(\cP)(X)=
\cup\{ C\in P ~|~ \exists s\in X. \: P(s)\tle C \}.
$$
Note that, for any $s\in \Sigma$, 
$\mu_\cP (s) =
\mu_\cP (P(s)) = \cup\{ C\in P ~|~ P(s)\tle C \}$. 
For preorder PRs, this map allows us to characterize the property of
being a simulation PR as follows. 

 \begin{theorem}
   \label{simpr}
   Let $\cP= \tuple{P,\tle}$ be a preorder PR. Then, $\cP$ is a
   simulation 
   iff
   \begin{enumerate}
     \item[{\rm (i)}] if $B\tle C$, $b\in B$ and $c\in C$ then $\ell(b)=\ell(c)$;
   \item[{\rm (ii)}] if $B\raee C$ and $B \tle D$ then $D\raee \mu_\cP(C)$;
   \item[{\rm (iii)}] for any $C\in P$, $P=\Split(P,\pre( \mu_\cP (C)))$.
   \end{enumerate}
 \end{theorem}
\begin{proof}
$(\Ra)$  Condition (i) clearly holds. 
Assume that $B\raee C$ and $B\tle D$. Hence, there exist
$b\in B$ and $c\in C$ such that $b\sra c$. Consider any state $d\in D$. 
Since $\cP$ is a
simulation and $P(b)\tle P(d)$, there exist some state $e$ 
such that $d\sra e$ and $C=P(c) \tle P(e)$. Hence, $D \raee \mu_\cP(C)$. 
Finally, if $C\in P$ and $x\in \pre(\mu_\cP(C))$ then there exists some
block $D\tre C$ and state $d\in D$ such that $x\sra d$. If
$y\in P(x)$ then since $P(x)=P(y)$, by reflexivity of $\tle$, 
we have that $P(x) \tle P(y)$, so that, since $\cP$
is a
simulation, there exists some state $e$ such that $y\sra e$ and $P(d) \tle P(e)$. 
Since $\tle$ is transitive, we have that $C\tle P(e)$. 
Hence,
 $y\in \pre(\mu_\cP(C))$. We have thus shown that $P(x)\subseteq
\pre(\mu_\cP(C))$, so that $P=\Split(P,\pre(\mu_\cP(C)))$. 
 
\noindent
$(\La)$ Let us show that $\Rel(\cP)$ is a simulation, i.e., 
if $P(s) \tle P(s')$ then: (a)~$\ell(s)=\ell(s')$; (b)~if $s\sra t$
then there exists $t'$ such that $s'\sra t'$ and $P(t) \tle P(t')$.  
Condition (a) holds by hypothesis (i). If $s \sra t$ then
$P(s) \sraee P(t)$ so that, by condition~(ii),  we have that 
$P(s') \sraee \mu_\cP
(P(t))$, namely there exists $s''\in P(s')$ such that
$s''\in \pre(\mu_\cP(P(t)))$. By condition (iii), $P(s'')
\subseteq \pre(\mu_\cP(P(t)))$, i.e., $s'\in
\pre(\mu_\cP(P(t)))$. Hence, there exists $t'$ such that $s'\sra t'$
and $P(t) \tle P(t')$. 
\end{proof}

\subsection{Partition and Relation Refiners}

By Theorem~\ref{simpr}, assuming that condition~(i) holds, 
there are two possible reasons for a PR
$\cP= \tuple{P,\tle}$ for not being a simulation:

 \begin{itemize}
 \item[{\rm (1)}] There exist $B,C,D\in P$ such that $B\raee C$, $B\tle D$,
   but $D\not\!\!\raee \mu_\cP (C)$; in this case we say that the
   block $C$ is a
   \emph{relation refiner} for $\cP$.
 \item[{\rm (2)}] 
 There exist $B,C\in P$ such that $B\cap\pre(\mu_\cP (C))\neq\varnothing$ 
   and $B\smallsetminus\pre(\mu_\cP (C))\neq \varnothing$; in this case we say that the
   block $C$
   is a \emph{partition refiner} for $\cP$.  
 \end{itemize}

\noindent
We therefore define:
\begin{itemize}
\item[] $\RRefiner(\cP) \ud \{C\in P~|~ C ~\text{is a relation refiner for}~
\cP\}$;
\item[] $\PRefiner(\cP) \ud \{C\in P~|~ C ~\text{is a partition refiner for}~
\cP\}$.
\end{itemize}
Accordingly, $\cP$ is defined to be 
relation or partition \emph{stable} when, respectively,
$\RRefiner (\cP) = \varnothing$ or $\PRefiner(\cP)=\varnothing$.   
Then, Theorem~\ref{simpr} can be read as follows: $\cP$ is a simulation iff $\cP$ 
satisfies condition~(i) and is both
relation and partition stable.

If $C\in \PRefiner (\cP)$ then
$P$ is first refined to $P' \ud  \Split(P,\pre(\mu_\cP(C)))$, i.e.\ $P$ 
is split w.r.t.\ the splitter $S= \pre(\mu_\cP(C))$. Accordingly,
the relation 
$\tle$ on $P$ is 
transformed into the following relation $\tle'$ defined on $P'$:

\renewcommand{\theequation}{\fnsymbol{equation}}
\setcounter{equation}{1}

\begin{equation} \label{eq1}
\tle' \;\:\ud\;\{(D,E)\in P'\times P' ~|~  
\parent_P(D) \tle \parent_P (E)\}
\end{equation}
Hence, two blocks $D$ and $E$ of the refined partition $P'$ 
are related by $\tle'$ if their parent blocks $\parent_P(D)$ and $\parent_P(E)$ 
in $P$ were related by $\tle$. Hence, if $\cP' = \tuple{P',\tle'}$ then
for all $D\in P'$, we have that $\mu_{\cP'}(D) = \mu_{\cP}(\parent_P(D))$. 
We will show that 
this refinement of $\tuple{P,\tle}$ is correct because if $B\in P$ is split
into $B\smallsetminus S$ and $B\cap S$ then all the states in $B\smallsetminus
S$ are not simulation equivalent to all the states in $B\cap S$.
Note that if $B\in P$ has been split into 
$B\cap S$
and $B\smallsetminus S$ then both $B\cap S \tle' B\smallsetminus S$ and 
$B\smallsetminus S \tle' B\cap S$ hold, and consequently $\cP'$ becomes
relation unstable.

On the other hand, if $\cP$ is partition stable and 
$C\in \RRefiner(\cP)$ then we will show that  
$\tle$ can be safely refined to the following relation $\tle'$:
\begin{gather}
\begin{aligned}\label{eq2}
\tle' &\ud \;\tle \smallsetminus \{ (B,D)\in P\times P~|~ B\raee C,\:
B \tle D,\:  D \not\!\!\raee \mu_\cP(C)\}\\
&=  \{ (B,D)\in P \times P~|~ B\tle D,\: 
 \big(B\raee C \Ra  D \raee \mu_\cP(C)\big)\}
\end{aligned}
\end{gather}
because if $(B,D)\in \;\tle\! \smallsetminus\! \tle'$ then all the states
in $D$ cannot simulate all the states in $B$.

\IncMargin{1.25em}
\LinesNumbered
\begin{algorithm}
\small
\PrintSemicolon
\SetAlgoVlined
\SetNlSkip{1.5em}
\SetAlTitleFnt{textsc}

\Indm
\FuncSty{{\rm $\mathrm{ESim} (\KwSty{PR}\;\tuple{P,\tle})\;\{$}}

\Indp
$\mathit{Initialize}()$;~$\PStabilize()$;~$\KwSty{bool} \PStable :=\RStabilize()$;~$\KwSty{bool} \RStable := \ctt$\;
\While{{\rm $\neg (\PStable \:\&\: \RStable)$}} {
\lIf{$\neg \PStable$}{\{$\RStable :=  \PStabilize()$;~
  $\PStable := \ctt$;\}}

\lIf{$\neg \RStable$}{\{$\PStable := \RStabilize()$;~
  $\RStable := \ctt$;\}}
}

\Indm
$\}$

\BlankLine
\BlankLine

\FuncSty{{\rm \KwSty{bool}~$\PStabilize ()\;\{$}}

\Indp
$P_{\mathrm{old}} := P$\;

\While{$\exists C\in \PRefiner(\cP)$}{
       $S:=\pre(\mu_\cP(C))$;
       $P:= \Split(S)$\;
       \lForAll{$(D,E)\in P\times P$}{
         $D \tle E \;:=\; \parent_P(D) \tle \parent_P(E)$\;
       }
       }

\creturn $(P=P_{\mathrm{old}})$\;

\Indm
$\}$

\BlankLine
\BlankLine

\FuncSty{{\rm $\KwSty{bool}~\RStabilize()\;\{$}}

\Indp

\tcp*[h]{{\rm \text{Precondition: $\PStable = \ctt$}}}

$\tle_{\mathrm{old}} \;:=\; \tle$;~ $\text{Delete} := \varnothing$\;

\While{$\exists C\in \RRefiner (\cP)$}{
      $\text{Delete} \; :=\; \text{Delete} 
      \cup \{ (B,D)\in P\times P~|~ B \tle D,\: B\raee C,\:  D \not\!\!\raee \mu_\cP(C)\}$\;     
   }
   $\tle \; :=\; \tle \smallsetminus \;\text{Delete}$\;

\creturn $(\tle \;=\; \tle_{\mathrm{old}})$\;

\Indm
$\}$

\caption{Logical Simulation Algorithm.}\label{basicsa}
\end{algorithm}

The above facts lead us to design 
a basic simulation algorithm $\ESim$ described 
in Figure~\ref{basicsa}. $\ESim$ maintains a PR $\cP=\tuple{P,\tle}$, 
which initially is $\tuple{P_\ell, \id}$ and is
iteratively refined as follows: 

\medskip
\noindent
$\PStabilize()$: If $\tuple{P,\tle}$ is not partition stable
then the partition $P$ is split for $\pre(\mu_\cP (C))$ 
as long as a partition refiner $C$ for $\cP$ exists, and when this happens
the relation $\tle$ is transformed to $\tle'$ as defined by (\ref{eq1}); at the end
of this process, we obtain a PR $\cP'=\tuple{P',\tle'}$ which is partition stable 
and if $P$ has been actually refined, i.e.\ $P' \prec P$ 
then the current PR $\cP'$ becomes relation unstable.

\medskip
\noindent
$\RStabilize()$: If $\tuple{P,\tle}$ is not relation stable
then the relation $\tle$ is refined to $\tle'$ 
as described by (\ref{eq2}) as long as a relation refiner for $\cP$ exists; hence, 
at the end of this refinement process 
$\tuple{P,\tle'}$ becomes relation stable but possibly partition unstable. 

\medskip
Moreover, the following properties of the current PR of $\ESim$ hold. 

\begin{lemma}\label{preorder} {In any run of\/} $\ESim$, the following
two conditions hold:
\begin{itemize}
\item[{\rm (i)}] If $\PStabilize()$ is
called on a partial order PR \mbox{$\tuple{P,\tle}$} 
then at the exit we obtain a PR $\tuple{P',\tle'}$ 
which is a preorder. 
\item[{\rm (ii)}]   If $\RStabilize()$ is
called on a preorder PR \mbox{$\tuple{P,\tle}$} 
then at the exit we obtain a PR $\tuple{P,\tle'}$ 
which is a partial order. 
\end{itemize}
\end{lemma}
\begin{proof}\let\qed\relax
Let us first consider $\PStabilize()$. Consider an input partial order PR
$\cP=\tuple{P,\tle}$, a splitter $S$ such that $P'=\Split(P,S)$
and let $\tle'$ be defined as in equation~(\ref{eq1}). 
Let us show that $\tuple{P',\tle'}$ is a preorder PR. 
 
\begin{description}
\item[{\rm (Reflexivity):}]
If $B\in P'$ then, as $\tle$ is reflexive, 
$P(B) \tle P(B)$ and thus $B\tle' B$.

\item[{\rm (Transitivity):}] Assume that $B,C,D\in P'$ and 
$B\tle' C$ and $C\tle' D$. Then, 
$P(B)\tle P(C)$ and $P(C)\tle P(D)$, so that by transitivity of $\tle$, 
$P(B)\tle  P(D)$. Hence, $B\tle' D$. 
\end{description}

\noindent 
Let us then take into account $\RStabilize()$, consider an input preorder PR
$\cP=\tuple{P,\tle}$ and let $\tuple{P,\tle'}$ be the output PR of $\RStabilize()$. 

\begin{description}
\item[{\rm (Reflexivity):}] If $B\in P$ then, by reflexivity of $\tle$, 
$B\tle B$. If 
$B\sraee C$, for some $C\in P$, then since $C\tle C$ by reflexivity of $\tle$, 
we have that 
$B\sraee \mu_\cP(C)$. Hence, $B\tle' B$.

\item[{\rm (Transitivity):}] Assume that $B\tle' C$ and $C\tle' D$. Then, 
$B\tle C$ and $C\tle D$, so that by transitivity of $\tle$, $B\tle D$. 
If $B\sraee E$ then, since $B\tle' C$, $C \sraee \mu_\cP(E)$. Hence,
there exists $F\in P$ such that $E\tle F$ and $C\sraee F$. 
Since $C\tle' D$, we have that $D \sraee \mu_\cP(F)$. Since
$\tle$ is transitive and $E\tle F$, $\mu_\cP(F) \subseteq \mu_\cP(E)$. 
Thus, we have shown that $B\sraee E$ implies $D\sraee \mu_\cP(E)$,
namely $B\tle' D$.

\item[{\rm (Antisymmetry):}] We observe that after calling $\PStabilize()$
on a partial order PR, 
antisymmetry can be lost because for any block $B$ which is split 
into $B_1=B\cap S$ and $B_2=B\smallsetminus S$, where $S=\pre(\mu_\cP(C))$, 
we have that $B_1 \tle B_2$ and
$B_2 \tle B_1$. In this case, $\RStabilize()$ removes the pair $(B\cap S,
B\smallsetminus S)$ from the relation $\tle$: in fact, while $B\cap S  \subseteq
\pre(\mu_\cP(C))$ and therefore $B\cap S\sraee E$, for some block $E\subseteq 
\mu_\cP(C)$, we have that $(B\smallsetminus S) \cap  \pre(\mu_\cP(C)) = \varnothing$,
so that, since $\mu_\cP(E)\subseteq \mu_\cP(C)$, 
$(B\smallsetminus S) \cap  \pre(\mu_\cP(E)) = \varnothing$, i.e., 
$B\smallsetminus S\not\!\!\sraee \mu_\cP(E)$, and therefore 
$B\cap S \not\tle'
B\smallsetminus S$. Hence, $\RStabilize()$ outputs a relation $\tle'$  which is
antisymmetric. \hfill\ensuremath{\Box}
\end{description}
\end{proof}

The main loop of $\ESim$ terminates when the current PR $\tuple{P,\tle}$ becomes both
partition and relation stable. By the above Lemma~\ref{preorder}, 
the output PR $\cP$ of $\ESim$ is a partial order, and hence a preorder,  
so that Theorem~\ref{simpr} can be applied to $\cP$ which 
then results to be a simulation PR. 
It turns out that this algorithm is correct, meaning
that the output PR $\cP$ actually represents the simulation preorder.

 \begin{theorem}[\textbf{Correctness}]
   \label{basiccorrect} Let $\Sigma$ be finite. 
$\ESim$ is correct, i.e., $\ESim$ terminates on any input and if $\tuple{P,\tle}$ is the output PR of
$\ESim$ on input $\tuple{P_\ell,\id}$ then for any $s,t\in \Sigma$, 
$s\leq t \;\Lra\; P(s) \tle P(t)$. 
 \end{theorem}
\begin{proof}
  Let us first note that $\ESim$ always terminates. In fact, if
  $\tuple{P,\tle}$ is the current PR at the beginning of some
  iteration of the while-loop of $\ESim$ and $\tuple{P',\tle'}$ is the
  current PR at the beginning of the next iteration then, since
  $\tuple{P',\tle'}$ is either partition or relation unstable, we have
  that either $P' \tl P$ or $P'=P$ and $\tle' \, \subsetneq \,
  \tle$. Since the state space $\Sigma$ is finite, at some iteration
  it must happen that $P'=P$ and $\tle' \, =\, \tle$ so that
  $\text{PStable} \:\&\: \text{RStable}=\ctt$.

  When $\ESim$ terminates, we have that $\RRefiner(\tuple{P,\tle})
  =\varnothing = \PRefiner(\tuple{P,\tle})$. Also, 
  let us observe that condition~(i)
  of Theorem~\ref{simpr} always holds for the current PR
  $\tuple{P,\tle}$ because the input PR $\tuple{P_\ell,\id}$ initially
  satisfies condition~(i) and this condition is clearly preserved at any
  iteration of $\ESim$. Furthermore, at the beginning, 
  we have that $\tuple{P,\tle}=
\tuple{P_\ell,\id}$ and this is trivially a partial
order. Thus, we can apply Lemma~\ref{preorder} for any call to $\PStabilize()$
and $\RStabilize()$, so that we obtain 
that  the output PR $\tuple{P,\tle}$ is a preorder. Hence, 
Theorem~\ref{simpr} can be applied to the output preorder 
PR $\tuple{P,\tle}$, which  is then a
  simulation. Thus, $\Rel(\tuple{P,\tle})\subseteq \Rsim$.

Conversely, let us show that if $\cP$ is the output PR of $\ESim$ then
$\Rsim \subseteq \Rel(\cP)$. This is shown as follows:
if $\cP$ is a
preorder PR such that $\Rsim \subseteq \Rel(\cP)$ and 
$\RStabilize()$ or $\PStabilize()$ are 
called on $\cP$ 
then at the exit we obtain a PR $\cP'$ 
such that $\Rsim \subseteq \Rel(\cP')$.

Let us first take into account $\RStabilize()$, consider an input preorder PR
$\cP=\ok{\tuple{P,\tle}}$ such that $\Rsim \subseteq \Rel(\cP)$,
and let $\cP'=\tuple{P,\tle'}$ be the output PR of $\RStabilize()$. We show that 
$\Rsim \subseteq \Rel(\cP')$, that is, for any $s,t\in \Sigma$, 
if $s\leq t$ then $P(s) \tle' P(t)$.  
By hypothesis, from $s\leq t$ we obtain $P(s) \tle P(t)$. Assume that
$P(s) \sraee C$, for some $C\in P$. Hence, $P(s) \sraee \mu_\cP(C)$. 
Since the PR $\cP$ is partition stable, we have that $P(s) \subseteq
\pre(\mu_\cP(C))$. Thus, there exists some $D\in P$ and $d\in D$ such
that $C\tle D$ and $s \sra d$. Therefore, since $t$ simulates $s$, there
exists some state $e$ such that $t\sra e$ and $d\leq e$. 
By hypothesis, from $d\leq e$ we obtain $D=P(d) \tle P(e)$. 
Hence, from $C\tle D$ and $D\tle P(e)$, since $\tle$ is transitive,
we obtain $C\tle P(e)$. Thus, $t\sraee \mu_\cP(C)$ and in turn $P(t) \sraee
\mu_\cP(C)$. We can thus conclude that $P(s) \tle' P(t)$. 

Let us now consider $\PStabilize()$. Consider an input preorder PR
$\cP=\tuple{P,\tle}$ (which, by Lemma~\ref{preorder}, actually
is a partial order PR) such that $\Rsim \subseteq \Rel(\cP)$.
Consider a splitter $S$ such that $P'=\Split(P,S)$
and let $\tle'$ be defined as in equation~(\ref{eq1}). 
Let $\cP'=\tuple{P',\tle'}$ and let us check that 
$\Rsim \subseteq \Rel(\cP')$, i.e., if $s\leq t$ then $P'(s)\tle' P'(t)$. 
By hypothesis, if $s\leq t$ then $P(s) \tle P(t)$. Moreover, 
by definition of $\tle'$ and since $P'\preceq P$, 
$P(s) \tle P(t)$ iff
$P'(s) \tle' P'(t)$. 

To sum up, we have shown that for the output PR $\tuple{P,\tle}$, 
$\Rsim = \Rel(\tuple{P,\tle})$, so that $s\leq t$ iff $P(s) \tle P(t)$.   
\end{proof}

\section{Efficient Implementation} 
\label{secimplementation}

\subsection{Data Structures}

$\ESim$ is implemented by relying on 
the following data structures.

\medskip
\noindent
\textbf{\emph{States:}} A state $s$ is represented by a record  that contains 
the list  $\post(s)$ of its  successors, a pointer $s$.block to the block $P(s)$
  that contains $s$ and a boolean flag used for marking purposes.  
  The whole state space $\Sigma$
  is represented as a doubly linked list of states. 
  $\{\post(s)\}_{s\in \Sigma}$ 
therefore represents the input transition system.  
  
\medskip
\noindent
\textbf{\emph{Partition:}}
The states of any block $B$ of the current partition
  $P$ are consecutive in the list $\Sigma$, so that $B$ is represented
  by two pointers begin and end: $B.\text{begin}$ is 
  the first state of $B$ in $\Sigma$ and $B.\text{end}$ is the
  successor of the last state of $B$ in $\Sigma$, i.e.,
  $B=[B.\text{begin},B.\text{end}[$.  Moreover, $B$ stores a
  boolean flag $B$.intersection and a block
  pointer
  $B$.brother whose meanings are as follows: after a
  call to $\mathit{Split}(P,S)$ for splitting $P$ w.r.t.\ a set of
  states $S$, if 
  $B_1 = B\cap S$ and $B_2 = B\smallsetminus S$, for
  some $B\in P$ that has been split by $S$ then 
  $B_1$.intersection $= \ctt$ and $B_2$.intersection $=\cff$, while 
$B_1$.brother points to $B_2$ and $B_2$.brother points to $B_1$. If instead
$B$ has not been split by $S$ then $B$.intersection
$=\cnull$ and $B$.brother $=\cnull$.
 Also, any block $B$ stores in $\Remove(B)$ a list
  of blocks of $P$, which is used by $\RStabilize()$, 
  and in $B.$preE the list of blocks $C\in P$ such
  that $C\raee B$.  Finally, any block $B$ stores  in
$B.\text{size}$ the size of $B$, in $B.\text{count}$ an integer counter bounded 
by $|P|$ which is used by $\PStabilize()$ and a pair
of boolean flags used for marking purposes.
The current partition $P$ is stored as a doubly linked list of blocks.  

\medskip
\noindent
\textbf{\emph{Relation:}}
The current relation $\tle$ on $P$ is stored as a
  resizable $|P|\times|P|$ boolean matrix. 
  Recall  \cite[Section~17.4]{cormen} that insert
  operations in a resizable array (whose capacity is doubled as
  needed) take amortized constant time and that a resizable matrix (or
  table) can be implemented as a resizable array of resizable arrays.  
  The boolean matrix $\tle$ is resized
  by adding a new entry to $\tle$, namely a new row and a new column,
  for any block $B$ that is split into two new blocks $B\smallsetminus
  S$ and $B\cap S$. The old entry $B$ becomes the entry for the
  new block $B\smallsetminus S$ while the new entry is used for the
  new block $B\cap S$. 
 
\medskip 
\noindent
\textbf{\emph{Auxiliary Data Structures:}} We store and maintain a
 resizable boolean matrix BCount and a resizable integer matrix Count,
 both indexed over $P$, whose meanings are as follows: 

\begin{itemize}
\item[] $\text{BCount}(B,C)\ud \left\{
  	\begin{array}{ll}
	1 & \mbox{ if $B\sraee C$} \\
	0 &\mbox{ if $B\not\!\!\sraee C$}
	\end{array}
	\right.$
\item[]
	$\text{Count}(B,C) \ud
      \textstyle{\sum}_{E\tre C} \text{BCount}(B,E)$
\end{itemize}

\noindent
Hence, $\text{Count}(B,C)$ stores the number of blocks $E$ 
such that $C\tle E$ and $B\sraee E$.  
The table Count allows to implement the test 
$B \nraee \pre(\mu_\cP (C))$ in constant time  as $\text{Count}(B,C)=0$.

\begin{algorithm}[t]
\footnotesize 
\SetAlgoVlined
\SetNlSkip{1.5em}
\SetAlTitleFnt{textsc}

\Indm
\FuncSty{{\rm $\textit{Initialize}()\;\{$} }

\Indp
\tcp*[h]{{\rm \text{Initialize $\text{BCount}$}}}

  \ForAll{$B\in P$}{
  	\lForAll{$C\in P$}{
    	    $\text{BCount}(B,C):=0$\;
        	}
        }

  \ForAll{$B\in P$}{ 
  	\ForAll{$x\in B$}{ 
		\ForAll{$y\in \post(x)$}{
        \lIf{{\rm $(\text{BCount}(B,y.\text{block})=0)$}}{$\text{BCount}(B,y.\text{block}):=1$\;}
      }
      }
      }

\tcp*[h]{{\rm \text{Initialize $\text{preE}$}}}

\textit{updatePreE}(); \tcp*[h]{{\rm \text{In Figure~\ref{updateFig}}}}

\BlankLine

\tcp*[h]{{\rm \text{Initialize $\text{Count}$}}}

\ForAll{$B\in P$}{ 
	\lForAll{$C\in P$}{
        $\text{Count}(B,C):=0$\;
        }
        }

\ForAll{$D\in P$}{ 
	\ForAll{{\rm $B\in D.$preE}}{
		\lForAll{$C\in P$ \KwSty{such that} $C\tle D$}{$\text{Count}(B,C)$++\;}
	}
	}

\tcp*[h]{{\rm \text{Initialize $\Remove$}}}

\ForAll{$C\in P$}{
	\ForAll{$B\in P$}{ 
		\ForAll{{\rm $D\in B$.preE}}{
			\lIf{{\rm $(\text{Count}(D,C)=0)$}}{$\Remove(C).\text{append}(D)$\;}	
	}
	}
	}
\Indm
$\}$

\caption{Initialization of data structures.}\label{initializeFig}
\end{algorithm}

\medskip
The data structures BCount, preE, Count and $\Remove$ 
are initialized by a function 
$\textit{Initialize()}$ at line~2 of $\ESim$, which is described in
Figure~\ref{initializeFig}.

\begin{algorithm}[t]
\footnotesize
\SetAlgoVlined
\SetNlSkip{1.5em}
\SetAlTitleFnt{textsc}

\Indm

\FuncSty{{\rm \clist{\cstate} $\pre\!\mu(\cblock~C)\;\{$} }

\Indp
\clist{\cstate} $S := \varnothing$\;

\ForAll{$x\in \Sigma$}{ 
	\ForAll{$y\in \post(x)$}{
    \lIf{{\rm $(C\tle y.\text{block} \;\&\; \text{unmarked}(x))$}}{
        $\{S.\text{append}(x)$;~ $\text{mark}(x)$;$\}$
      }
    }
    }
\lForAll{$x\in S$}{$\text{unmark}(x)$\;} 

\creturn $S$\;

\Indm
$\}$

\BlankLine

\clist{\cblock} \FuncSty{$\mathit{Split}($\clist{\cstate} $S) \;\{$}

\Indp
\clist{\cblock} split\;
\lForAll{$B\in P$}{$B$.intersection $:=\cnull$\;}

\ForAll{$x \in S$}{
  \If{{\rm $(x.\text{block}.\text{intersection} = \cnull)$}}{
    $x.\text{block}.\text{intersection} := \cff$\;
    \cblock  ~$B$ $:=$ new \cblock;~ 
    $x$.block.brother $:= B$\;
    $B$.brother $:=$ $x$.block;~ 
    $B$.intersection $:=\ctt$\;
    split.append($x$.block)\;
  }
  move $x$ in list $\Sigma$ from $x$.block at the end of $B$\;
  \If(\tcp*[h]{{\rm $x$.block $\subseteq S$}}){{\rm $(x$.block $= \varnothing)$}}{
   $x$.block.begin $:= B$.begin;~
   $x$.block.end $:= B$.end\;
    $x$.block.brother $:= \cnull$;~ 
    $x$.block.intersection $:= \cnull$\;
    split.remove($x$.block);~ delete $B$\;
 }
}

\creturn split\;

\Indm
$\}$
 
\caption{{\rm Split algorithm.}}\label{SplitFig}
\end{algorithm}

 \subsection{Partition Stability}
 \label{secfindrefiner}

Our implementation of $\ESim$ will exploit
the following logical characterization of partition refiners. 

 \begin{theorem}
   \label{refiner}
   Let $\tuple{P,\tle}$ be a partial order PR. Then,
   $\PRefiner(\ok{\tuple{P,\tle}})\neq \varnothing$ 
   iff there exist $B,C\in P$ such that the
   following three conditions hold: 
   \begin{itemize}
   \item[{\rm (i)}] $\ok{B\raee C}$;
   \item[{\rm (ii)}] for any $C' \in P$, if $C\tl C'$ then 
     $\ok{B\not\!\!\raee C'}$;
   \item[{\rm (iii)}] $B\not\subseteq\pre(C)$.
   \end{itemize}  
 \end{theorem}
\begin{proof}
Let $\cP = \tuple{P,\tle}$. 

\noindent
($\La$) 
{}From condition~(i) we have that 
$B\cap \pre(\mu_\cP(C))\neq \varnothing$. From conditions~(ii) and (iii), 
$B\not\subseteq
\pre(\mu_\cP(C))$. Thus, $C\in \PRefiner(\cP)$.

\noindent
($\Ra$)
Assume that $\PRefiner(\cP)\neq \varnothing$. Since $\tuple{P,\tle}$
is a partial order, we consider a partition refiner 
$C\in \max (\PRefiner(\cP))$  which is maximal w.r.t.\ the partial order
$\tle$. Since $C$ is a partition refiner, there exists some $B\in P$
such that 
$B\cap \pre(\mu_\cP(C))\neq \varnothing$ and $B\not\subseteq
\pre(\mu_\cP(C))$. If $C'\in P$ is such that $C\tl C'$ then 
$C'$ cannot be a partition refiner because $C$ is a maximal partition refiner. 
Hence, if $B\sraee C'$ then $B\subseteq \pre(\mu_\cP(C'))$, because
$C'$ is not a partition refiner, so that, since $\pre(\mu_\cP(C'))
\subseteq \pre(\mu_\cP(C))$,
 $B\subseteq \pre(\mu_\cP(C))$, which is a contradiction. Hence, for
 any $C'\in P$ if $C\tl C'$ then  
 $B\not\!\!\sraee C'$. Therefore, from $B\cap
 \pre(\mu_\cP(C))\neq \varnothing$ we obtain that $B\sraee C$. Moreover,
 from $B\not\subseteq
\pre(\mu_\cP(C))$ we obtain that $B\not\subseteq \pre(C)$. 
\end{proof}

\begin{algorithm}[t]
\footnotesize
\SetAlgoVlined
\SetNlSkip{1.5em}
\SetAlTitleFnt{textsc}

\Indm
\FuncSty{{\rm $\KwSty{bool}~\PStabilize ()\;\{$}}

\Indp
\clist{\cblock} $\textit{split} := \varnothing$\;

\While{{\rm $(C:=\mathit{FindPRefiner}()) \neq \cnull$)}}{
  \clist{\cstate} $S := \pre\!\mu(C)$;~ $\textit{split} := \textit{Split}(S)$\;
  \textit{updateRel}($\textit{split}$);~  
  \textit{updateBCount}($\textit{split}$);~
  \textit{updatePreE}()\;
  \textit{updateCount}($\textit{split}$);~
  \textit{updateRem}($\textit{split}$)\;
}

\creturn $(\textit{split}= \varnothing)$\;

\Indm
$\}$

\BlankLine
\BlankLine

\FuncSty{{\rm $\KwSty{Block}~\mathit{FindPRefiner} ()\;\{$}}

\Indp

\ForAll{$B\in P$} {
  \clist{\cblock} $p := \mathit{Post}(B)$\;
  \ForAll{{\rm $C \in p$}}{
    \lIf{{\rm $(\Count(B,C)=1)$}}{\creturn $C$\;}
  }
}
\creturn \cnull\;

\Indm
$\}$

\BlankLine
\BlankLine

\FuncSty{{\rm \clist{\cblock} $\mathit{Post}(\cblock~B)\;\{$} }

\Indp

\clist{\cblock} $p := \varnothing$\;

  \ForAll{$b \in B, $}{
    \ForAll{$c \in \post(b)$}{
      \cblock~ $C:=c.\text{block}$\;
      \lIf{{\rm $\text{unmarked1}(C)$}}{
      $\{$mark1($C$);
      $C.\text{count}=0$;
       $p$.append($C$);$\}$
      }
      
      \lIf{{\rm $\text{unmarked2}(C)$}}{
        $\{$mark2($C$);
        $C$.count++;$\}$
      }
    }
    \lForAll{{\rm $C \in p$}}{
      unmark2($C$)\;
    }    
  }
  \ForAll{{\rm $C \in p$}}{ 
  unmark1($C$)\;
  \lIf{{\rm $(C.\text{count} = B.\text{size})$}}
     {$p$.remove($C$);\}}
     }
  
  \creturn $p$\;

\Indm
$\}$

\caption{{\rm $\PStabilize()$} Algorithm.}\label{pstabilizeFig}
\end{algorithm}

Notice that this characterization of partition refiners 
requires that the current PR is a partial order relation and, by Lemma~\ref{preorder},
for any call
to $\PStabilize()$,
this is  actually guaranteed by  the  $\ESim$ algorithm.

The algorithm in Figure~\ref{pstabilizeFig} is an implementation of 
the $\PStabilize()$ function that relies on Theorem~\ref{refiner}
and on the above data structures. The function 
$\mathit{FindPRefiner()}$
implements the conditions of Theorem~\ref{refiner}: it 
returns a partition refiner for the current PR $\cP=\tuple{P,\tle}$ 
when this exists, otherwise it returns a null
pointer. Given a block $B\in P$, the function $\mathit{Post}(B)$
returns a list of blocks $C\in P$ that satisfy conditions~(i) and~(iii) of
Theorem~\ref{refiner}, i.e., those blocks $C$ such that $B\raee C$ and
$B\not\subseteq \pre(C)$. This is accomplished through the counter 
$C.\text{count}$ that at the exit of the for-loop at lines~18-23 in Figure~\ref{pstabilizeFig} stores 
the number of states in $B$ having (at least) an outgoing transition to $C$,
i.e., $C.\text{count} = |B\cap \pre(C)|$. Hence, we have that:
$$
B\sraee C \text{~and~} 
B\not\subseteq \pre(C) \,\Lra\,
1 \leq C.\text{count} < B.\text{size}.
$$
Then, for any candidate partition refiner $C\in \mathit{Post}(B)$, it remains
to check condition~(ii) of Theorem~\ref{refiner}. This condition is
checked in $\mathit{FindPRefiner}()$ 
by testing whether $\text{Count}(B,C)=1$: this is correct
because $\text{Count}(B,C)\geq 1$ holds since $C\in \mathit{Post}(B)$ and therefore 
$B\sraee C$, so that
$$\text{Count}(B,C) =1 \;\text{~iff~}\; 
\forall C' \in P. C\tl C' \Ra B\not\!\!\sraee C'.$$
Hence, if $\text{Count}(B,C)=1$ holds at line~13 of $\mathit{FindPRefiner}()$, 
by Theorem~\ref{refiner}, $C$ is a partition refiner. 
Once a partition refiner $C$ has been returned by $\mathit{Post}(B)$,
$\PStabilize()$ splits the current partition $P$ w.r.t.\ the splitter
$S= \pre(\mu_\cP (C))$ by calling the function 
$\mathit{Split}(S)$, 
updates the relation $\tle$ as defined by equation~(\ref{eq1}) in
Section~\ref{secbasicsa} by calling $\mathit{updateRel}()$, 
updates the 
data structures BCount, preE, Count and $\Remove$, and 
then check again whether a partition refiner exists. 
At  the exit of the main while-loop of $\PStabilize()$, the current 
PR $\tuple{P,\tle}$ is partition stable.

$\PStabilize()$ calls the functions $\pre\!\mu()$ and 
$\mathit{Split}()$ in Figure~\ref{SplitFig}.
Recall that the states of a block $B$ of 
$P$ are consecutive in the list of states $\Sigma$, so that $B$ is represented
as $B=[B.\text{begin},B.\text{end}[$.
The implementation of 
$\mathit{Split}(S)$ is 
quite standard (see e.g.\
\cite{gv90,rt10}): this is based on a linear 
scan of the states in $S$ and for each state in $S$ performs
some constant time operations. Hence, 
$\mathit{Split}(S)$ takes $O(|S|)$
time. Also, $\mathit{Split}(S)$ returns the list $\textit{split}$ of blocks $B\smallsetminus S$ such that 
$\varnothing \subsetneq B\smallsetminus S \subsetneq B$ (i.e., $B$.intersection $=\cff$).
Let us remark that a call $\mathit{Split}(S)$ may affect the ordering of the
states in the list 
$\Sigma$ because states are moved 
from old blocks to newly generated
blocks. 

We will show that the overall time complexity of $\PStabilize()$ 
along a whole
run of $\ESim$ is in 
$O(|\Psim||\sra|)$.

\begin{algorithm}
\footnotesize 
\SetAlgoVlined
\SetNlSkip{1.5em}
\SetAlTitleFnt{textsc}

\Indm
\FuncSty{{\rm $\textit{updateRel}($\clist{\cblock} $\textit{split})\;\{$} }

\Indp
\lForAll{$B\in \textit{split}$}{addNewEntry$(B)$ in matrix $\tle$\;}

\ForAll{$B \in P$}{ 
	\ForAll{$C \in \textit{split}$} {
  \lIf{{\rm $(B$.intersection $= \ctt)$}} {
	$B \tle C := B.\text{brother} \tle C.\text{brother}$\;
	}
	\lElse{$B \tle C := B \tle C.\text{brother}$\;}
	}
	}
	
\ForAll{$C \in P$}{ 
	\ForAll{$B \in \textit{split}$} {
  \lIf{{\rm $(C$.intersection $= \cff)$}} {
	$B \tle C := B.\text{brother} \tle C$\;
	}
	}
	}
	
\Indm
$\}$

\BlankLine

\FuncSty{{\rm $\textit{updateBCount}($\clist{\cblock} $\textit{split})\;\{$} }

\Indp
\lForAll{$B\in \textit{split}$}{addNewEntry$(B)$ in matrix Count\;}

  \ForAll{$B\in P$}{ 
  	\ForAll{$x\in B$}{
		\lForAll{$y\in \post(x)$}{
        $\text{BCount}(B,y.\text{block}):=0$\;
        }
      }
}

  \ForAll{$B\in P$}{
  \ForAll{$x\in B$}{
  	\ForAll{$y\in \post(x)$}{
        \lIf{{\rm $(\text{BCount}(B,y.\text{block})=0)$}}{$\text{BCount}(B,y.\text{block}):=1$\;}
      }
      }
      }
  
\Indm
$\}$

\BlankLine

\FuncSty{{\rm $\textit{updatePreE}()\;\{$} }

\Indp

\lForAll{$B \in P$}{$B.\text{preE}:=\varnothing$\;}

\ForAll{$B\in P$}{ 
	\ForAll{$x \in B$}{
		\lForAll{$y\in \post(x)$}{
      $\{\text{unmark}(B)$;~
      $y.\text{block}.\text{preE}.\text{append}(B)$;$\}$
    }
    }
    }

\ForAll{$C\in P$}{
  \ForAll{{\rm $B\in C.\text{preE}$}}{
    \lIf{{\rm $\text{unmarked}(B)$}}{$\text{mark}(B)$\;}
    \lElse{$C.\text{preE}.\text{remove}(B)$\;}
  }
  \lForAll{{\rm $B\in C.\text{preE}$}}{$\text{unmark}(B)$\;}
}

\Indm
$\}$

\BlankLine

\FuncSty{{\rm $\textit{updateRem}($\clist{\cblock} $\textit{split})\;\{$} }

\Indp

\lForAll{$B\in \textit{split}$}{
  $\Remove(B) := \Remove(B.\text{brother})$\;
  }

\Indm
$\}$

\caption{Update functions.}\label{updateFig}
\end{algorithm}

\subsection{Updating Data Structures}\label{uds}

In the function $\PStabilize()$, after calling $\Split(S)$, 
firstly we need to update the boolean matrix that stores
the relation $\tle$ in accordance with definition ($\ref{eq1}$)
in Section~\ref{secbasicsa}. After that, since both $P$ and $\tle$ are changed
we need to update the data structures BCount, preE, Count and $\Remove$.   
The implementations of the functions $\textit{updateRel}()$, $\textit{updateBCount}()$, 
$\textit{updatePreE}()$ and $\textit{updateRem}()$ are quite straightforward and are
described in Figure~\ref{updateFig}.

The function $\textit{updateCount}()$ 
is in Figure~\ref{updateCountFig} and 
deserves special
care in order to design a time efficient 
implementation. The core of the $\textit{updateCount}()$ 
algorithm 
follows Hopcroft's 
``process the smaller half'' principle \cite{hop71} 
for updating the integer matrix
Count. Let $P'$ be the partition which is obtained by splitting the
partition $P$ w.r.t.\ the splitter $S$. 
Let $B$ be a block of $P$ that has been split into $B\cap S$ and
$B\smallsetminus S$. Thus, we need to update
$\text{Count}(B\cap S, C)$ and 
$\text{Count}(B\smallsetminus S, C)$ for any $C\in P'$
 by knowing $\text{Count}(B,\parent_P(C))$. Let us first observe that 
after lines~3-10 of $\mathit{updateCount}()$, we have that for any $B,C\in P'$, 
$\text{Count}(B,C) = \text{Count}(\parent_P(B),\parent_P(C))$.
Let $X$ be the block in $\{B\cap S,B\smallsetminus S\}$
with the smaller size, and let $Z$ be the other block, so that $|X|\leq |B|/2$
and $|X|+|Z|=|B|$. Let $C$ be any block in $P'$. 
We set $\text{Count}(X,C)$ to $0$, while 
$\text{Count}(Z,C)$ is left unchanged, namely 
$\text{Count}(Z,C) = \text{Count}(B,C)$. 
We can correctly update both $\text{Count}(Z,C)$ and $\text{Count}(X,C)$ 
by just scanning all the outgoing transitions from
$X$. In fact, if $x\in X$, $x\sra y$ and the block
$P(y)$ is scanned for the first time then for all 
$C\tle P(y)$, $\text{Count}(X,C)$
is incremented by $1$,  
while if \mbox{$Z \not\!\!\sraee P(y)$}, i.e.\ $\text{BCount}(Z,P(y))=0$, 
then $\text{Count}(Z,C)$ is decremented by $1$. 
The correctness of this procedure goes as follows: 
\begin{itemize}
\item[{\rm (1)}] At the end, $\text{Count}(X,C)$ is clearly correct
because its value has been re-computed from scratch.
\item[{\rm (2)}] At the end, $\text{Count}(Z,C)$ is correct because
 $\text{Count}(Z,C)$ initially stores the value $\text{Count}(B,C)$, and 
 if there exists some block $D$ such that $C\tle D$, $B\sraee D$ whereas 
 \mbox{$Z\not\!\!\sraee D$}~---~this is correctly 
 implemented at line~28 as $\text{BCount}(Z,D)=0$, since
 the date structure BCount is up to date~---~then necessarily $X\sraee D$, because $B$ has been split into 
 $X$ and $Z$, so that $D=P(y)$ for some $y\in \post(X)$, namely
 $D$ has been taken into account by some increment 
 $\text{Count}(X,C)$++ 
 and consequently $\text{Count}(Z,C)$ is decremented by $1$ at line~28.
\end{itemize} 

\begin{algorithm}[t]
\footnotesize 
\SetAlgoVlined
\SetNlSkip{1.5em}
\SetAlTitleFnt{textsc}

\Indm

\tcp*[h]{{\rm \text{Precondition: $\text{BCount}$ and preE are updated with the current PR}}} 

\FuncSty{{\rm $\textit{updateCount}($\clist{\cblock} $\textit{split})\;\{$} }

\Indp
\lForAll{$B\in \textit{split}$}{addNewEntry$(B)$ in matrix Count\;}

\ForAll{$B \in P$}{ 
	\ForAll{$C \in \textit{split}$} {
  \lIf{{\rm $(B$.intersection $= \ctt)$}} {
	$\text{Count}(B,C) := \text{Count}(B.\text{brother},C.\text{brother})$\;
	}
	\lElse{$\text{Count}(B,C) := \text{Count}(B,C.\text{brother})$\;}
	}
	}
	
\ForAll{$C \in P$}{ 
	\ForAll{$B \in \textit{split}$} {
  \lIf{{\rm $(C$.intersection $= \cff)$}} {
	$\text{Count}(B,C) := \text{Count}(B.\text{brother},C)$\;
	}
	}
	}

\lForAll{$C\in P$}{unmark$(C)$\;}

\ForAll{{\rm $B \in \textit{split}$}}{
  \tcp*[h]{{\rm \text{Update $\text{Count}(B,\cdot)$ and $\text{Count}(B.\text{brother},\cdot)$}}}

  \KwSty{Block} $X$, $Z$\;
  \eIf{{\rm ($B.\text{size} \leq B.\text{brother}.\text{size}$)}}{
    \{$X:=B$; $Z := B.\text{brother}$;\}
 }
{
    \{$X:=B.\text{brother}$;~ $Z := B$;\}
  }

    \ForAll{$C\in P$}{$\text{Count}(X,C):=0$\;
    \tcp*[h]{{\rm $~\text{Count}(Z,C):=\text{Count}(B,C)$;}}
    }

  \ForAll{$x \in X$}{
  	\ForAll{$y \in \post(x)$}{
		\If{{\rm unmarked$(y.\text{block})$}}{
    	mark$(y.\text{block})$\;
        \ForAll{$C\!\in\! P$ \KwSty{such that} $C\tle y.\text{block}$}{
        Count($X$,$C$)++\;
        \lIf{{\rm $(\text{BCount}(Z,y.\text{block})\!=\!0)$}}
        {Count($Z$,$C$)$\,$--$\,$--;$\!\!$}
        }               
      }
      }
      }
    \tcp*[h]{{\rm \text{$\!$For all $D\!\not\in\! \{B,B.\text{brother}\}$, update $\!\text{Count}(D,\cdot)$}}}

  \ForAll{{\rm $D\in X$.preE}}{
  	\If{{\rm $(D\!\neq\! X \,\&\, D\!\neq\! Z \,\&\, \text{BCount}(D,Z)\!=\!1)$}} {
		\lForAll{$C\in P$ \KwSty{such that} $C\tle X$}{
    $\text{Count}(D,C)$++\; 
  }
  }
  }
  }

\Indm
$\}$

\caption{{\rm \textit{updateCount}$()$} function.}\label{updateCountFig}
\end{algorithm}

\noindent
Moreover, if some block $D\in P'\smallsetminus\{B\cap S, 
B\smallsetminus S\}$ is such that both 
$D\sraee X$ and $D\sraee Z$ hold then 
for all the blocks $C\in P$ such that $C\tle X$ 
(or, equivalently, $C\tle Z$), we need to increment $\text{Count}(D,C)$
by $1$. This is done at lines~30-32 by relying on the updated date structures 
preE and BCount. 

Let us observe that the time complexity of 
a single call of $\textit{updateCount}(\mathit{split})$ 
is 
$$
|P|\big(|\mathit{split}| + 
{\sum_{X\in \textit{split}}}  \big(|\{(x,y)~|~ x\in X, y\in \Sigma, x\sra y\}|
+|\{(X,D)~|~ D\in P, X\sraee D\}|\big)\big).
$$
Hence, let us  calculate the overall time complexity of
$\textit{updateCount}()$. 
If $X$ and $X'$ are two blocks that are scanned in two different 
calls of $\textit{updateCount}$ and $X'\subseteq X$ then $|X'|\leq |X|/2$. 
Consequently, any  transition $x\sra y$  at line~23 and 
$D\sraee X$ at line~30
can be 
scanned in some call of $\textit{updateCount}()$ at most
$\log_2 |\Sigma|$ times.
Thus, the overall time
complexity of $\textit{updateCount}()$ is in
$O(|\Psim||\sra|\log|\Sigma|)$.

\IncMargin{0.25em}
\begin{algorithm}
\label{rstabilizeFig}
\footnotesize
\SetAlgoVlined
\SetNlSkip{1.5em}
\SetAlTitleFnt{textsc}

\Indm
\FuncSty{{\rm $\KwSty{bool}~\RStabilize ()\;\{$}}

\Indp
\tcp*[h]{$\mu_\cP^{\mathrm{in}} := \mu_\cP;$}

\lForAll{$C\in P$}{$\{\oldRemove(C) := \Remove(C);~ \Remove(C)=\varnothing;\}$}

\KwSty{bool} $\text{Removed} := \cff$\;

\ForAll{$C\in P$ \KwSty{such that} $\oldRemove(C)\neq \varnothing$}{

\tcp*[h]{{\rm $\!$\text{Invariant (Inv): $\forall C\in P. \Remove(C) =\{D\in P~|~ 
D\,\sraee \mu_\cP^{\mathrm{in}} (C),\: D\not\!\!\sraee \mu_\cP (C)\}$}}}

      \ForAll{{\rm $B\in C.\text{preE}$}}{
      	\ForAll{{\rm $D\in \oldRemove(C)$}}{
		\If{$(B\tle D)$}{     
            $B\tle D := \cff$;~
            $\text{Removed} := \ctt$\;
            \tcp*[h]{{\rm update Count and $\Remove$}}
            
            \ForAll{{\rm $F\in D$.preE}}{
              $\text{Count}(F,B) := \text{Count}(F,B)\!-\!1$\;
              \If(\tcp*[h]{$F\sraee \mu_\cP^{\mathrm{in}}(B)\;\&\; F\not\!\!\sraee\mu_\cP(B)$}){{\rm $(\text{Count}(F,B)=0 \;\& \Remove(B)=\varnothing)$}}{
                {$\Remove(B).\text{append}(F)$\;}
            }
            }
            }
}
}
}

\Indm
$\}$

\caption{{\rm $\RStabilize()$} algorithm.}
\end{algorithm}

\subsection{Relation Stability}
 The logical procedure $\RStabilize()$ in Figure~\ref{basicsa}
is implemented by the  algorithm in
 Figure~\ref{rstabilizeFig}.  
Let \mbox{$\cP^{\mathrm{in}}=\tuple{P,\tle^{\mathrm{in}}}$} 
be the current PR when calling $\RStabilize()$. 
For each relation refiner $C\in P$, the function $\RStabilize()$ must iteratively 
refine the initial relation $\tle^{\mathrm{in}}$ in accordance with
equation~(\ref{eq2}) in Section~\ref{secbasicsa}. 
Hence, if $B\sraee C$, $B\tle D$ and $D\ok{\not\!\!\sraee} \mu_{\cP^{\mathrm{in}}} (C)$, 
the entry $B\tle D$ of the boolean matrix that represents the relation $\tle$ must
be set to $\cff$. Thus, the idea is to store and
incrementally maintain for each block $C\in P$ a list
$\Remove(C)$ of blocks $D\in P$ such that:
\begin{itemize}
\item[{\rm (A)}] If $C$ is a relation refiner for $\cP^{\mathrm{in}}$ then $\Remove(C)\neq
\varnothing$;  
\item[{\rm (B)}] If $D\in \Remove(C)$ then necessarily
\mbox{$D\not\!\!\sraee \mu_\cP^{\mathrm{in}} (C)$}. 
\end{itemize}
It turns out that 
$C$ is a relation refiner for $\cP^{\mathrm{in}}$ iff there exist blocks $B$ and $D$
such that $B\sraee C$, $D\in \Remove(C)$ and $B\tle D$. Hence, the set of blocks 
$\Remove(C)$ 
is reminiscent of the set of states $\mathit{remove}(s)$ used 
in Henzinger et al.'s~\cite{hhk95} simulation algorithm, since each pair 
$(B,D)$ which must be removed from the relation $\tle$ is 
such that $D\in \Remove(C)$, for some block $C$. 
\newline
\indent  
Initially, namely at the first call
of $\RStabilize()$ by $\ESim$, $\Remove(C)$ is set by the function
$\textit{Initialize}()$ 
to $\{D\in P~|~D\sraee \Sigma,\, D\ok{\not\!\!\sraee}\mu_\cP(C)\}$.
Hence, 
$\RStabilize()$ scans all the blocks in the current partition $P$ and selects
those blocks $C$ such that $\Remove(C)\neq \varnothing$, 
which are therefore candidate to be relation refiners. 
Then, by scanning all the blocks $B\in C.$preE and $D\in \Remove(C)$, 
if $B\tle D$ holds then the entry $B\tle D$ must be set to $\cff$. 
However, the removal of the pair $(B,D)$ from the current relation $\tle$ may affect
the function $\mu_\cP$. This is avoided by making a copy $\oldRemove(C)$ of all the
$\Remove(C)$'s at the beginning of $\RStabilize()$ and then using this copy. 
During the main for-loop of $\RStabilize()$, 
$\Remove(C)$ must satisfy the following invariant property: 
$$
\text{(Inv):}~~\forall C\in P. \Remove(C) = \{D\in P~|~
D\,\sraee \mu_\cP^{\mathrm{in}} (C),\:
 D\not\!\!\sraee \mu_\cP (C)\}.
$$
This means that at the beginning of $\RStabilize()$, any $\Remove(C)$ is set
to empty, and after the removal of a pair $(B,D)$ from $\tle$, since
$\mu_\cP(B)$ has changed,  
we need: (i) to update the matrix Count, for all the entries
$(F,B)$ where $F\sraee D$, and (ii) to
check 
if there is some block $F$ such that $F\not\!\!\sraee \mu_\cP(B)$, because
any such $F$ must be added to $\Remove(B)$ in order to maintain  the 
invariant property (Inv).

\subsection{Complexity}
The time complexity 
of the algorithm $\ESim$ relies on the following key properties:
\begin{itemize}
\item[{\rm (1)}] The overall number of partition refiners
found by $\ESim$ 
is  in $O(|\Psim|)$.  
Moreover, the overall number of newly
generated blocks by the splitting operations performed by calling
$\Split(S)$ at line~4 of $\PStabilize()$
is in $O(|\Psim|)$. 
In fact, let $\{P_i\}_{i\in [0,n]}$ be the
sequence of different partitions computed by $\ESim$ where $P_0$ is the
initial partition $P_{\ell}$, $P_n$ is the final partition $\Psim$ and
for all $i\in [1,n]$, $P_i$ is the partition
after the $i$-th call to $\Split(S)$, so that 
$P_{i} \prec P_{i-1}$. The number of new
blocks which are produced by a call $\Split(S)$ that refines $P_i$ to
$P_{i+1}$ is $2(|P_{i+1}| - |P_i|)$. Thus, the overall number of newly
generated blocks is $\sum_{i=1}^{n} 2(|P_{i}| - |P_{i-1}|) =
2(|P_{\mathrm{sim}}|-|P_{\ell}|)\in  O(|\Psim|)$.
\item[{\rm (2)}]
The invariant (Inv) 
of the sets $\Remove(C)$ guarantees the following property: 
if $C_1$ and $C_2$ are two blocks that are selected by the
for-loop at line~5 of $\RStabilize()$ in two different calls
of $\RStabilize()$, and $C_2 \subseteq C_1$
(possibly $C_1 = C_2$) then $(\cup \!\Remove(C_1))\cap (\cup
\!\Remove(C_2))=\varnothing$.
\end{itemize}

\begin{theorem}\label{complexity}
  $\ESim$ runs in 
  $O(|P_{\mathrm{sim}}|^2 \log |\Psim| + |\Sigma|\log |\Sigma|)$-space and
  $O(|P_{\mathrm{sim}}||\sra|\log |\Sigma|)$-time.
\end{theorem}
\begin{proof}
\textbf{\emph{Space Complexity.}}
The input transition system is represented by 
the $\post$ relation, so that the size of $\post$ is not taken into account in
the space complexity of $\ESim$.  The doubly linked list of states take 
$O(|\Sigma| \log |\Sigma|)$  
while the pointers $s$.block take $O(|\Sigma| \log |\Psim|)$. 
The partition $P$ and the pointers stored in
eack block of $P$ overall take $O(|\Psim|\log|\Sigma|)$. 
The binary relation $\tle$ takes $O(|\Psim|^2)$. The auxiliary data
structures  $\Remove$, preE and BCount 
overall take $O(|\Psim|^2)$, 
while the integer matrix Count takes  $O(|\Psim|^2\log |\Psim|)$.
Hence, the overall bit space complexity for storing the above data structures 
is  $O(|\Psim|^2 \log |\Psim| + |\Sigma|\log |\Sigma|)$. 

\medskip
\noindent
\textbf{\emph{Time Complexity.}}
The time complexity bound of $\ESim$ is shown by the following points. 
\begin{itemize}
\item[{\rm (A)}] The initialization function $\mathit{Initialize}$ takes 
$|P|^2 + |\sra| + |P| |\{(B,D)~|~B,D\in P,\: B\sraee D\}|$ time. 
Observe that $|P|\leq |\Psim| \leq |\sra|$ so that 
the time complexity of $\mathit{Initialize}$ 
is in $O(|\Psim||\sra|)$. 

\item[{\rm (B)}] 
A call to $\pre\!\mu(C)$ takes $O(|\sra|)$ time.
A call to $\mathit{Split}(S)$ takes $|S|$ time. Since $S$ is
returned by $\pre\!\mu(C)$, $|S|\leq |\sra|$ holds  so that
the time complexity of a
call to $\mathit{Split}(S)$ is in $O(|\sra|)$.
A call to $\textit{Post}(B)$ takes  $|\{(b,c)~|~b\in B,\, c\in \Sigma,\, b\sra c\}|$ time, so that
a call to $\textit{FindPRefiner}$ takes $O(|\sra|)$ time.  
Moreover, let us observe that $\textit{FindPRefiner}$ returns $\cnull$ just once,
because when $\textit{FindPRefiner}$ returns $\cnull$ the current PR of $\ESim$ 
is both 
partition and relation stable and therefore $\ESim$ terminates and outputs
that PR. 
Consequently, since, by point~(1) above, the overall number of partition refiners 
is in $O(|\Psim|)$, the overall number of function calls 
for $\textit{FindPRefiner}$ is in  $O(|\Psim|)$ and, in turn, 
the overall time complexity of $\textit{FindPRefiner}$ is in $O(|\Psim||\sra|)$ time. 
Also, 
the overall time complexity of $\pre\!\mu(C)$ and $\mathit{Split}(S)$
is in   $O(|\Psim||\sra|)$.

\item[{\rm (C)}] Let us observe that the calls
$\textit{updateRel}(\mathit{split})$ and $\textit{updateRem}(\mathit{split})$
take $O(|P||\mathit{split}|)$ time,  while $\textit{updatePreE}()$ and
$\textit{updateBCount}(\mathit{split})$
take $O(|\sra|)$ time.  Since the overall number of calls for these functions 
is in $O(|\Psim|)$ and since $\sum_{i\in \text{Iterations}} |\mathit{split}_i|$ 
is in $O(|\Psim|)$, 
 it turns out that their overall time complexity is
in $O(|\Psim|(|\Psim| + |\sra|))$, so that, since $|\Psim|\leq |\sra|$, 
it is in $O(|\Psim||\sra|)$.  Moreover, 
as already shown in Section~\ref{uds}, the overall time
complexity of $\textit{updateCount}(\mathit{split})$ is in
$O(|\Psim||\sra|\log|\Sigma|)$.

\item[{\rm (D)}] 
Hence, by points~(B) and (C), the overall time complexity of
$\PStabilize()$ is in $O(|\Psim||\sra|\log|\Sigma|)$. 

\item[{\rm (E)}] 
Let $C\in P_{\mathrm{in}}$ be some block of
the initial partition and let $\langle C_i\rangle_{i\in I_C}$, for some
set of indices $I_C$, be a
sequence of blocks selected by the for-loop  at line~5 of
$\RStabilize()$ such that: (a)~for any $i\in I_C$,
$C_i\subseteq C$ and (b)~for any $i$, $C_{i+1}$ has been selected after $C_i$ 
and $C_{i+1}$ is contained in
$C_i$. Observe that $C$ is
the parent block in $P_{\mathrm{in}}$ of all the $C_i$'s. Then, by the property~(2)
above, it
turns out that the
corresponding sets in $\{\cup\! \Remove(C_i)\}_{i\in I_C}$ are pairwise
disjoint so that $\sum_{i\in I_C} |\Remove(C_i)| \leq |\Psim|$. This property
guarantees that if $D\in \oldRemove(C_i)$ at line~8
then for all the blocks $D'\subseteq D$
and for any $j \in I_C$ such that $i<j$,  
$D' \not\in \oldRemove(C_j)$.
Moreover, if the test $B \tle D$ at line~9 is true for some iteration $k$, 
so that $B\tle D$ is set to $\cff$, 
then for all the blocks $D'$ and $B'$ such that $D'\subseteq D$
and $B'\subseteq B$ the test $D'\tle B'$ will be always false for all the
iterations which follow $k$.   
{}From these observations, we derive that 
the overall time complexity of the 
code of the for-loop at lines~7-10 is  
$\sum_C \sum_{i\in I_C} \sum_{B\sraee C} |\Remove(C_i)|   \leq
|P_{\mathrm{sim}}||\{(B,C)~|~B,C\in \Psim, B\sraee C\}|\leq |\Psim||\sra|$. 
Moreover, the overall time complexity of the code of the for-loop 
at lines 12-15 is $\sum_B \sum_D  
\sum_{F \sraee D} 1 \leq  |P_{\mathrm{sim}}||\{(F,D)~|~F,D\in \Psim, F\sraee D\}|\leq
|\Psim||\sra|$. 
We also observe that the overall time complexity of
the for-loop at line~3 of $\RStabilize()$ is in $O(|\Psim|^2)$. 
Thus, the overall time complexity of $\RStabilize()$ is in $O(|\Psim|(|\Psim|+|\sra|))$,
so that, since $|\Psim|\leq |\sra|$, 
it is in $O(|\Psim||\sra|)$.   
\end{itemize}
Summing up, by points~(A), (D) and (E), 
we have shown that the overall time complexity of $\ESim$ is in
$O(|\Psim||\sra|\log|\Sigma|)$. 
\end{proof}

\section{Conclusion and Further Work}
We have introduced a new algorithm, called ESim, 
for efficiently computing the simulation preorder which: 
(i)~reaches the space bound of the simulation algorithm GPP~\cite{gpp02,gpp03}~---~which has 
the best space complexity~---~ 
while  significantly improving its time bound; 
(ii)~significantly 
improves the space bound of the simulation algorithm RT~\cite{rt07,rt10}~---~which has 
the best time complexity~---~while closely approching its time bound. 
Moreover, the space complexity of ESim is quasi-optimal, meaning 
that it differs only for logarithmic factors from the size of the output. 

We see a couple of interesting avenues for further work. 
A first natural question arises: can the time complexity of 
ESim be further improved and reaches the time complexity of RT? 
This would require to eliminate the multiplicative factor $\log |\Sigma|$ 
from the time complexity of ESim and, presently, this seems 
to us quite hard 
to achieve. More in general, it would be interesting 
to investigate whether 
some lower space and time bounds can be stated 
for the simulation preorder problem. 
Secondly, ESim is designed for Kripke structures. 
While an adaptation of a simulation algorithm from 
Kripke structures  to labeled transition systems 
(LTSs) can be conceptually simple, unfortunately such a shift may lead
to some loss in both space and time complexities, as argued in~\cite{cec13}. We mention 
the works \cite{abdulla,hs09} and 
\cite{mar11} that 
provide simulation algorithms for LTSs by adapting, respectively,
RT and GPP. It is  thus worth investigating whether and how ESim can be 
efficiently 
adapted to work with  LTSs.  
 
\paragraph*{\rm \textbf{Acknowledgements.}} We acknowledge the contribution of Francesco Tapparo
to a preliminary stage of this research which was informally 
presented in \cite{rt09}.  This work was partially supported by 
Microsoft Research SEIF 2013 Award and by the University of
Padova under the project BECOM.


\begin{thebibliography}{99}
\small

\bibitem{abdulla}
P.A.~Abdulla, A,~Bouajjani, L.~Hol{\'{\i}}k, L.~Kaati and T.~Vojnar.
\newblock Computing simulations over tree automata.
\newblock In \emph{Proc.\ 14th Int.\ Conf.\ on
Tools and Algorithms for the Construction and Analysis of Systems 
(TACAS'08)}, LNCS~4963, pp.~93--108,
2008.


\bibitem{bloom89}
B.~Bloom. 
\newblock \emph{Ready simulation, bisimulation, and the semantics of CCS-like languages}.
\newblock Ph.D. thesis, Massachusetts Institute of
Technology, 1989.

\bibitem{bp95}
B.~Bloom and R.~Paige. 
\newblock Transformational design and implementation of a new
efficient solution to the ready simulation problem.
\newblock \emph{Sci.\ Comp.\ Program.}, 24(3):189-220, 1995.

\bibitem{bg03}
D.~Bustan and O.~Grumberg.
\newblock Simulation-based minimization.
\newblock \emph{ACM Trans.\ Comput.\ Log.},
4(2):181-204, 2003.

\bibitem{cec13}
G.~C{\'e}c{\'e}.  
\newblock Three simulation algorithms for labelled transition systems. 
\newblock Preprint cs.arXiv:1301.1638, arXiv.org, 2013. 

\bibitem{cgp99}
E.M.~Clarke, O.~Grumberg and D.A.~Peled.
\newblock \emph{Model Checking}.
\newblock The {M}{I}{T} Press, 1999.

\bibitem{cps93}
R.~Cleaveland, J.~Parrow and B.~Steffen. 
\newblock 
The concurrency workbench: A semantics based tool
for the verification of concurrent systems. 
\newblock \emph{ACM Trans.\ Program.\ Lang.\ Systems}, 15(1):36-72, 1993. 

\bibitem{cormen}
T.H.~Cormen, C.E.~Leiserson, R.L.~Rivest and C.~Stein.
\newblock \emph{Introduction to Algorithms}.
\newblock The {M}{I}{T} Press and McGraw-Hill, 2nd ed., 2001.


\bibitem{CRT11}
S.~Crafa, F.~Ranzato and F.~Tapparo.
\newblock 
Saving space in a time efficient simulation algorithm. 
\newblock \emph{Fundam.\ Inform.}, 108(1-2):23-42, 2011. 

\bibitem{gpp02}
R.~Gentilini, C.~Piazza and A.~Policriti.
\newblock Simulation as coarsest partition problem. 
\newblock In {\em Proc.\ Int.\ Conf.\ on Tools and Algorithms for the 
Construction and Analysis of Systems (TACAS'02)}, LNCS 2280, 
pp.~415-430, 2002. 

\bibitem{gpp03}
R.~Gentilini, C.~Piazza and A.~Policriti.
\newblock From bisimulation to simulation: coarsest partition problems.
\newblock \emph{J.\ Automated Reasoning}, 31(1):73-103, 2003. 



\bibitem{GP08}
R.~van~Glabbeek and B.~Ploeger.
\newblock Correcting a space-efficient simulation algorithm.
\newblock In \emph{Proc.\ 20th Int.\ Conf.\ on
Computer Aided Verification (CAV'08)}, LNCS~5123,
pp.~517-529, 2008.

\bibitem{gv90}
J.F.~Groote and F.~Vaandrager.
\newblock An efficient algorithm for branching bisimulation and
stuttering equivalence.
\newblock In \emph{Proc.\ 17th ICALP}, LNCS~443,
pp.~626-638,  1990.


\bibitem{hhk95}
M.R.~Henzinger, T.A.~Henzinger and P.W.~Kopke.
\newblock Computing simulations on finite and infinite graphs.
\newblock In \emph{Proc.\ 36th IEEE FOCS}, 453-462, 1995.

\bibitem{hs09}
L.~Hol{\'{\i}}k and J.~{\u{S}}im\'a{\u{c}}ek. 
\newblock Optimizing an LTS-simulation algorithm.
\newblock \emph{Computing and Informatics}, 29(6+):1337-1348, 2010. 

\bibitem{hop71}
J.E.~Hopcroft.
\newblock A {$n \log n$} algorithm for minimizing states in a finite
automaton. 
\newblock In Z.~Kohavi and A.~Paz eds., 
\emph{Theory of Machines and Computations}, pp.~189-176, Academic Press, 
1971.

\bibitem{mar11}
J.~Markovski.
\newblock Saving time in a space-efficient simulation algorithm. 
\newblock In \emph{Proc.\ 11th Int.\ Conf.\ on Quality
Software}, pp.~244-251, IEEE, 2011. 

\bibitem{pt87}
R.~Paige and R.E.~Tarjan.
\newblock Three partition refinement algorithms.
\newblock \emph{SIAM J.\ Comput.}, 16(6):973-989, 1987

\bibitem{ran-mfcs13}
F.~Ranzato. 
\newblock 
A more efficient simulation algorithm on {K}ripke structures. 
\newblock In {\em Proceedings of the 38th International Symposium on Mathematical Foundations of Computer Science (MFCS'13)}, LNCS~8087, pp.~753-764, Springer, 2013. 


\bibitem{rt07}
F.~Ranzato and F.~Tapparo. 
\newblock 
A new efficient simulation equivalence algorithm. 
\newblock In {\em Proc.\ 22nd Annual IEEE Symp.\ on Logic in Computer Science (LICS'07)}, pp.~171-180, 2007. 

\bibitem{rt09}
F.~Ranzato and F.~Tapparo. 
\newblock  A time and space efficient simulation algorithm. 
\newblock Short talk at
\emph{24th Annual IEEE Symposium on Logic in Computer Science (LICS'09)}, 2009.  


\bibitem{rt10}
F.~Ranzato and F.~Tapparo.
\newblock An efficient simulation algorithm based on abstract interpretation.
\newblock \emph{Inf.\ Comput.}, 208(1):1-22, 2010. 


\bibitem{TC01}
L.~Tan and R.~Cleaveland.
\newblock Simulation revisited.
\newblock In \emph{Proc.\ 
7th Int.\ Conf.\ on Tools and Algorithms for the 
Construction and Analysis of Systems 
(TACAS'01)}, LNCS~2031, pp.~480-495,
2001.


\end{thebibliography}
\end{document}